\documentclass{IEEEtran}

\usepackage{graphicx}
\usepackage{amsmath,amssymb,amsthm}
\usepackage{mathtools}
\usepackage{float}
\usepackage{xcolor,colortbl}
\usepackage{color}
\usepackage{fp}
\usepackage{tikz}
\usetikzlibrary{shapes.misc,shadows}
\usepackage{multirow}
\usepackage{xcolor,colortbl}
\usepackage{caption}
\usepackage{graphicx}
\usepackage[linesnumbered,ruled]{algorithm2e}
\usepackage{epstopdf}
\usepackage{pgfplots}
\pgfplotsset{compat=1.7}
\usepackage{caption}
\usepackage{subcaption}
\newtheorem{lemma}{Lemma}

\usepackage{url}
\usepackage[noadjust]{cite}
\usepackage{setspace}
\usepackage{gensymb}
\usepackage[ruled]{algorithm2e}
\graphicspath{{./Figures/}}
\DeclareMathOperator{\sinc}{sinc}
\SetKwInOut{kwvar}{Variables}
\SetKwInOut{kwin}{Inputs}
\newcommand{\nosemic}{\renewcommand{\@endalgocfline}{\relax}}
\newcommand{\dosemic}{\renewcommand{\@endalgocfline}{\algocf@endline}}
\let\oldnl\nl
\newcommand{\nonl}{\renewcommand{\nl}{\let\nl\oldnl}}
\newcommand{\comment}[1]{}

\usepackage[abbreviations, automake]{glossaries-extra}
\makeglossaries 
\newabbreviation{OMA}{OMA}{orthogonal multiple access}
\newabbreviation{3GPP}{3GPP}{$3^{rd}$ Generation Partnership Project}
\newabbreviation{NR}{NR}{New Radio}
\newabbreviation{NOMA}{NOMA}{Non-orthogonal multiple access}
\newabbreviation{BS}{BS}{base station}
\newabbreviation{IRS}{RIS}{reconfigurable intelligent surface}
\newabbreviation{los}{LoS}{line of sight}
\newabbreviation{uspa}{USPA}{uniform square planar array}
\newabbreviation{SIC}{SIC}{successive interference cancellation}
\newabbreviation{ASR}{ASR}{achievable sum rate}
\newabbreviation{EE}{EE}{energy efficiency}
\newabbreviation{MPA}{MPA}{maximum \gls{ASR} power allocation}
\newabbreviation{EEPA}{EEPA}{energy efficient power allocation}
\newabbreviation{SINR}{SINR}{signal to interference plus noise ratio}
\newabbreviation{CSI}{CSI}{channel state information}
\newabbreviation{SRM}{SRM}{sum rate maximisation}
\newabbreviation{aoa}{AoA}{angle of arrival}
\newabbreviation{aod}{AoD}{angle of departure}
\newabbreviation{AUP}{AUP}{adaptive user pairing}
\newabbreviation{FPA}{FPA}{fair power allocation}
\newabbreviation{LTE}{LTE}{Long Term Evolution}
\newabbreviation{QoS}{QoS}{Quality of Service}
\newabbreviation{D2D}{D2D}{device to device}
\newabbreviation{5G}{5G}{fifth-generation}

\newacronym{h_i}{\ensuremath{\mathbf{h}_i}}{Channel Coefficient from $i^{th}$ user to \gls{IRS}}
\newacronym{h_r}{\ensuremath{\mathbf{h}_R}}{Channel Coefficient from \gls{IRS} to \gls{BS}}
\newacronym{beta}{\ensuremath{\beta_i}}{Channel gain from $i^{th}$ user to \gls{IRS}}
\newacronym{alpha}{\ensuremath{\alpha}}{Channel gain from \gls{IRS} to \gls{BS}}
\newacronym{psiIa}{\ensuremath{\psi_I^a}}{Azimuth angle of arrival at the \gls{IRS}}
\newacronym{psiIe}{\ensuremath{\psi_I^e}}{Elevation angle of arrival at the \gls{IRS}}
\newacronym{psiBa}{\ensuremath{\psi_B^a}}{Azimuth angle of arrival at the \gls{BS}}
\newacronym{psiBe}{\ensuremath{\psi_B^e}}{Elevation angle of arrival at the \gls{BS}}
\newacronym{phiIa}{\ensuremath{\phi_I^a}}{Azimuth angle of departure at the \gls{IRS}}
\newacronym{phiIe}{\ensuremath{\phi_I^e}}{Elevation angle of departure at the \gls{IRS}}
\newacronym{h_r^H}{\ensuremath{\mathbf{h}^H_R}}{Hermitian of the Channel Coefficient from IRS to BS}
\newacronym{lambda}{\ensuremath{\lambda}}{Wavelength of the signal}
\newacronym{theta}{\ensuremath{\mathbf{\Theta}}}{Reflection Matrix of IRS without considering imperfect phase compensation}
\newacronym{Theta}{\ensuremath{\tilde{\mathbf{\Theta}}}}{Reflection Matrix of IRS with imperfect phase compensation}
\newacronym{thetak}{\ensuremath{\theta_k}}{Phase shift introduced by $k^{th}$ element at the \gls{IRS}}
\newacronym{thetakcap}{\ensuremath{\hat{\theta}_k}}{Phase error at $k^{th}$ element at \gls{IRS}}
\newacronym{P_t}{\ensuremath{P_t}}{Available power budget}
\newacronym{s_i}{\ensuremath{s_i}}{Data symbol transmitted by $i^{th}$ user}
\newacronym{n_i}{\ensuremath{n_i}}{Noise effecting $i^{th}$ user}
\newacronym{gammaoma}{\ensuremath{\gamma_i^{\text{OMA}}}}{SINR of $i^{th}$ OMA user}
\newacronym{I}{\ensuremath{I}}{Interference power}
\newacronym{sigma^2}{\ensuremath{\sigma^2}}{Noise variance}
\newacronym{alpha_1}{\ensuremath{\alpha_1}}{Fraction of power allocated to the strong user}
\newacronym{alpha_2}{\ensuremath{\alpha_2}}{Fraction of power allocated to the weak user}
\newacronym{s_1}{\ensuremath{s_1}}{Data symbol transmitted by the strong user}
\newacronym{s_2}{\ensuremath{s_2}}{Data symbol transmitted by the weak user}
\newacronym{gamma_1}{\ensuremath{\gamma_1^{\text{NOMA}}}}{SINR of the strong user}
\newacronym{gamma_2}{\ensuremath{\gamma_2^{\text{NOMA}}}}{SINR of the weak user}
\newacronym{delta}{\ensuremath{\delta}}{Maximum possible value of phase noise}
\newacronym{h_1}{\ensuremath{\mathbf{h}_1}}{Channel coefficient of strong user}
\newacronym{h_2}{\ensuremath{\mathbf{h}_2}}{Channel coefficient of weak user}
\newacronym{csi_i}{\ensuremath{\gamma_i^{\text{CSI}}}}{CSI of $i^{th}$ user}
\newacronym{csi_1}{\ensuremath{\gamma_1^{\text{CSI}}}}{CSI of strong user}
\newacronym{csi_2}{\ensuremath{\gamma_2^{\text{CSI}}}}{CSI of weak user}
\newacronym{R_1}{\ensuremath{R_1^{\text{NOMA}}}}{Rate of strong user}
\newacronym{R_2}{\ensuremath{R_2^{\text{NOMA}}}}{Rate of weak user}
\newacronym{R1bar}{\ensuremath{\bar{R}_1}}{Minimum data rate requirement of strong user}
\newacronym{R2bar}{\ensuremath{\bar{R}_2}}{Minimum data rate requirement of weak user}
\newacronym{alpha1mpa}{\ensuremath{\alpha_1^{\text{MPA}}}}{Fraction of power allocated to the strong user in \gls{MPA}}
\newacronym{alpha2mpa}{\ensuremath{\alpha_2^{\text{MPA}}}}{Fraction of power allocated to the weak user in \gls{MPA}}
\newacronym{alpha1eepa}{\ensuremath{\alpha_1^{\text{EEPA}}}}{Fraction of power allocated to the strong user in \gls{EEPA}}
\newacronym{alpha2eepa}{\ensuremath{\alpha_2^{\text{EEPA}}}}{Fraction of power allocated to the weak user in \gls{EEPA}}
\newacronym{alpha2lb}{\ensuremath{{\alpha_{2_{\text{LB}}}}}}{Lower bound on $\alpha_2$}
\newacronym{alpha2ub}{\ensuremath{{\alpha_{2_{\text{UB}}}}}}{Upper bound on $\alpha_2$}
\newacronym{delubmpa}{\ensuremath{\gls{delta}_{\text{UB}}^{\text{MPA}}}}{Upper bound on imperfect phase compensation in \gls{MPA}}
\newacronym{alphalb}{\ensuremath{\alpha_{2_{\text{LB}}}}}{Lower bound on \gls{alpha_2}}
\newacronym{eta}{\ensuremath{\eta}}{}
\newacronym{kappa}{\ensuremath{\kappa}}{}
\newacronym{delub1}{\ensuremath{\gls{delta}_{\text{UB}_1}}}{}
\newacronym{delub2}{\ensuremath{\gls{delta}_{\text{UB}_2}}}{}
\newacronym{delubeepa}{\ensuremath{\gls{delta}_{\text{UB}}^{\text{EEPA}}}}{Upper bound on imperfect phase compensation in \gls{EEPA}}

\usepackage{booktabs}
\usepackage{calc}
\usepackage{tabularx}
\newlength\maxlength\newlength\thislength
\newglossarystyle{mystyle}
{

\renewcommand*{\glsgroupheading}[1]{}

}
\begin{document}
\bstctlcite{IEEEexample:BSTcontrol}
\title{\LARGE Spectral and Energy Efficient User Pairing for RIS-assisted Uplink NOMA Systems with Imperfect Phase Compensation}
\author{\IEEEauthorblockN{Kusuma Priya P., Pavan Reddy M., Abhinav Kumar} \\
\IEEEauthorblockA{Department of Electrical Engineering, Indian Institute of Technology Hyderabad, India 502285\\Email:\{ee20mtech11007, ee14resch11005\}@iith.ac.in, abhinavkumar@ee.iith.ac.in}
\vspace{-0.6cm}}
\maketitle
\begin{abstract}
Non-orthogonal multiple access (NOMA) is considered a key technology for improving the spectral efficiency of fifth-generation (5G) and beyond 5G cellular networks. NOMA is beneficial when the channel vectors of the users are in the same direction, which is not always possible in conventional wireless systems. With the help of a reconfigurable intelligent surface (RIS), the base station can control the directions of the channel vectors of the users. Thus, by combining both technologies, the RIS-assisted NOMA systems are expected to achieve greater improvements in the network throughput. However, ideal phase control at the RIS is unrealizable in practice because of the imperfections in the channel estimations and the hardware limitations. This imperfection in phase control can have a significant impact on the system performance. Motivated by this, in this paper, we consider an RIS-assisted uplink NOMA system in the presence of imperfect phase compensation. We formulate the criterion for pairing the users that achieves minimum required data rates. We propose adaptive user pairing algorithms that maximize spectral or energy efficiency. We then derive various bounds on power allocation factors  for the paired users. Through extensive simulation results, we show that the proposed algorithms significantly outperform the state-of-the-art algorithms in terms of spectral and energy efficiency.
\end{abstract}
\IEEEpeerreviewmaketitle 
\vspace{-0.3cm}
\section{Introduction}
\label{sec:Intro} 
In the \gls{5G} and beyond \gls{5G} communication systems, significant improvements are expected in terms of spectral efficiency, energy conservation, massive connectivity, and latency requirements. \gls{NOMA} is considered as a key multiple access technique to enhance the spectral efficiency of future cellular networks~\cite{intro1}.  Various \gls{NOMA} transmission schemes have been evaluated for possible consideration of \gls{NOMA} in \gls{5G} uplink scenario~\cite{38812}. In \gls{NOMA}, multiple users are allocated with the same time and frequency resources to achieve multi-fold improvement in the network throughputs. However, the users are multiplexed in either power domain (power-domain \gls{NOMA}) or code domain (code-domain \gls{NOMA}), which requires the receiver to employ a \gls{SIC} and decode the data~\cite{intro3}. Similar to \gls{NOMA}, \gls{IRS} is considered as a key technology to improve the spectral efficiency of the cellular networks~\cite{vincentpoor}. An \gls{IRS} is a two-dimensional planar array consisting of low-cost reflecting antenna elements that can modify the amplitude and phase of the signal. Thus, by using \gls{IRS}, the \gls{BS} can perform beamforming and transmit the signal in the desired direction. 

Unlike the spatial multiplexing in \gls{OMA}, NOMA is beneficial in situations where the channel vectors of users are in the same direction~\cite{vincentpoor}. This is not always possible in conventional wireless systems, whereas, in the case of \gls{IRS}-assisted systems, the \gls{BS} can control the direction of the user channel vectors by tuning the \gls{IRS}~\cite{vincentpoor,IRS5}. For these reasons, \gls{IRS}-assisted \gls{NOMA} systems have been widely considered to achieve significant improvements in the network performance~\cite{vincentpoor,base}. However, the ideal phase control is difficult to achieve in practice for various reasons like hardware impairments, channel estimation errors, etc. These imperfections in the phase control degrade the achievable spectral and energy efficiencies in the network. Further, in \gls{NOMA} systems, the achievable data rates are significantly dependent on the user pairing, and hence, while pairing the users in \gls{IRS}-assisted \gls{NOMA} systems, the network operator has to consider the imperfections in the phase compensation. Otherwise, the expected improvements in the network throughputs will not be realized in practice. For the aforementioned reasons, investigating the effect of imperfection in phase at the \gls{IRS} is crucial to achieve optimum system performance.

In~\cite{IRS3,IRS4}, the authors have proposed various channel estimation techniques for \gls{IRS}-assisted wireless systems. 
In~\cite{srm,IRSNOMA5}, the authors have formulated the sum-rate maximization as an optimization problem, and then, derived a near-optimal solution for \gls{IRS}-assisted uplink \gls{NOMA} systems.
However, 
limited works in the literature have considered the imperfection in phase for the downlink \gls{OMA} and \gls{NOMA} systems~\cite{base,model,mouni}. To the best of our knowledge, none of the existing works considered the imperfections in the phase compensation while pairing the users in \gls{IRS}-assisted uplink \gls{NOMA} systems. Motivated by this, we present the following key contributions in this paper. 
\begin{itemize}
\item For \gls{IRS}-assisted uplink \gls{NOMA} systems, we derive bounds on imperfection in the phase compensation to achieve minimum required data rates.
\item We propose user pairing algorithms that maximize spectral efficiency and energy efficiency, respectively.
\item We derive bounds and define the power allocation factors for the paired users.
\item Through extensive simulation results, we show that the proposed algorithms significantly outperform the state-of-the-art algorithms.
\end{itemize}

The rest of the paper is organized as follows. 
The system model is presented in Section~\ref{sec:model}.  In Section~\ref{sec:algo}, we propose the adaptive user pairing algorithms, derive the bounds on the power allocation factors, and formulate the criterion for pairing the users to maximize spectral and energy efficiencies. In Section~\ref{sec:NR}, we present the simulation results for various scenarios. Section~\ref{sec:conclusion} presents some concluding remarks and possible future works.

\vspace{-1cm}
\section{System Model}
\vspace{-0.1cm}
\label{sec:model}
\begin{figure}
\centering 
\includegraphics[scale=0.4]{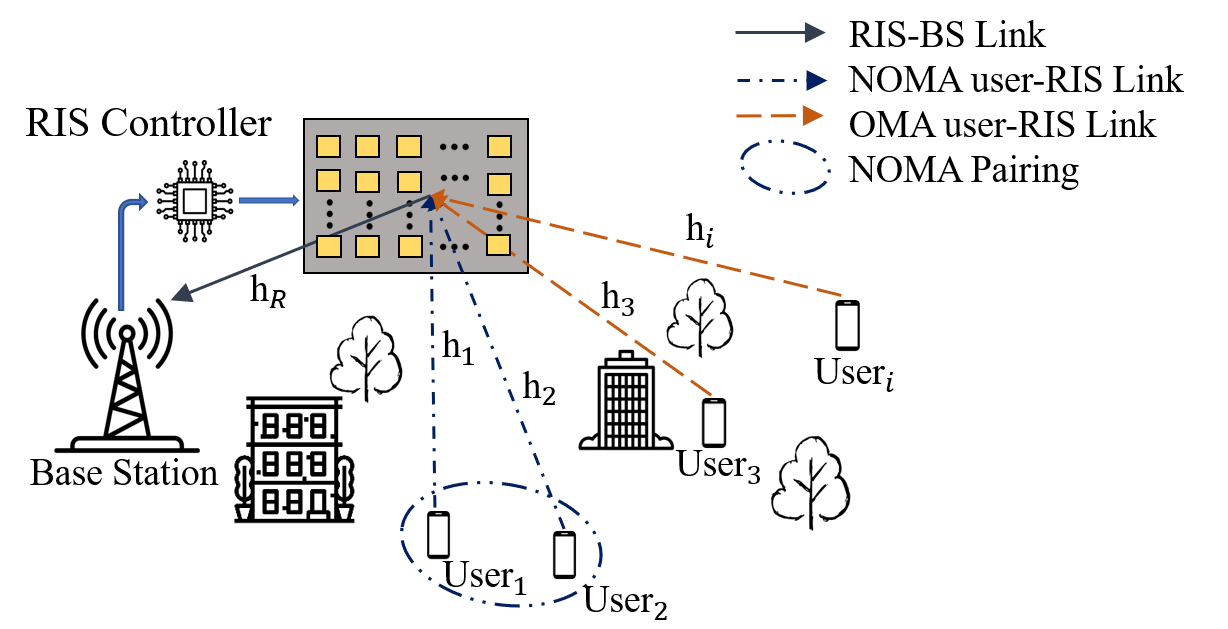}
\caption{System model.}
\vspace{-0.6cm}
\label{fig:SysModel}
\end{figure}
We consider an uplink scenario with $M$ antennae at the \gls{BS}, $N$ reflective antenna elements at the \gls{IRS}, and a single antenna at each user as shown in Fig.~\ref{fig:SysModel}. The direct link between the \gls{BS} and user is assumed to be blocked by the obstacles such as buildings, trees, or human body which is a likely scenario in case of mmwave communications~\cite{model}. Further,
the channel coefficients from $i^{th}$ user to \gls{IRS} and \gls{IRS} to \gls{BS} are defined as \gls{h_i} and \gls{h_r}, respectively and are formulated as follows~\cite{model}
\begin{align}
    \gls{h_i}&=\gls{beta}\mathbf{a}_N\left(\gls{psiIa},\gls{psiIe}\right),\label{eqn:hi}\\    
    \gls{h_r}&=\gls{alpha}\mathbf{a}_N\left(\gls{phiIa},\gls{phiIe}\right)\mathbf{a}_M^H\left(\gls{psiBa},\gls{psiBe}\right),\label{eqn:hr}
\end{align}
where \gls{beta} and \gls{alpha} denote the channel gains from $i^{th}$ user to \gls{IRS} and \gls{IRS} to \gls{BS}, respectively, \gls{phiIa} and \gls{phiIe} represent the \gls{aod} in azimuth and elevation at the \gls{IRS}, respectively, \gls{psiBa} and \gls{psiBe} represent \gls{aoa} in azimuth and elevation at the \gls{BS}, respectively, \gls{psiIa} and \gls{psiIe} represent the \gls{aoa} in azimuth and elevation at the \gls{IRS}, respectively, and 
$\mathbf{a}_X(\ohm^a,\ohm^e)$ represents the array response vector. 
A \gls{uspa} with $X$ antenna elements, has $\sqrt{X}$ elements in both horizontal and vertical directions, and thus, the array response vector is defined as 
\begin{align}
    \mathbf{a}_X(\ohm^a,\ohm^e)&=\begin{bmatrix}1\\ 
    \vdots\\e^{j\frac{2\pi d}{\gls{lambda}}(x\sin{\ohm^a}\sin{\ohm^a}+y\cos{\ohm^e})}\\ \vdots\\ 
    e^{j\frac{2\pi d}{\gls{lambda}}((\sqrt{X}-1)(\sin{\ohm^a}\sin{\ohm^a})+(\sqrt{X}-1)\cos{\ohm^e})}\end{bmatrix}^T,\nonumber
\end{align}
where, $d$ is the equispaced elemental distance, \gls{lambda} is the signal wavelength, $0\le x, y\le\sqrt{X}-1$ are the indices of the \gls{uspa} elements in the horizontal and vertical directions, respectively.

We denote the reflection matrix at the \gls{IRS} as \gls{theta} and define it as follows~\cite{vincentpoor}.
\begin{align}
 \gls{theta}&=\text{diag}(e^{j\theta_1},\hdots,e^{j\theta_k},\hdots,e^{j\theta_N}),\nonumber
\end{align}
where, \gls{thetak} $(1\le k \le N)$ is the phase shift introduced by $k^{th}$ reflecting element at the \gls{IRS}. However, note that the ideal phase control is difficult to achieve in practice because of various reasons like hardware limitations, imperfect channel state information, etc. For this reason, we consider imperfections in the phase control and define the practical reflection matrix as follows.
\begin{align}
 \gls{Theta}&=\text{diag}(e^{j\tilde{\theta}_1},\hdots,e^{j\tilde{\theta}_k},\hdots,e^{j\tilde{\theta}_N}),\nonumber
\end{align}
where, $\tilde{\theta}_k=\gls{thetak}+\gls{thetakcap}$ with \gls{thetakcap} being the error in phase control at $k^{th}$ antenna element. Further, we assume  \gls{thetakcap} to be uniformly distributed over $[-\gls{delta},\gls{delta}]$, $\gls{delta}\in[0,\pi)$.
Thus, in the case of \gls{OMA} transmission, the signal received from the $i^{th}$ user at the \gls{BS} is formulated as
\begin{equation}
    y_i^{\text{OMA}}=\gls{h_r^H}\gls{Theta}\gls{h_i}\gls{P_t}\gls{s_i}+\gls{n_i}, \nonumber
\end{equation}
where $\{ . \}^\text{H}$ denotes the Hermitian of a matrix, \gls{P_t} is the available transmit power at each user, \gls{s_i} is the data symbol transmitted by the $i^{th}$ user, and \gls{n_i} denotes the thermal noise. The \gls{SINR} of an $i^{th}$ user in an \gls{IRS}-assisted \gls{OMA} system is formulated as
\begin{equation}
    \gls{gammaoma}=\frac{\gls{P_t}{\|\gls{h_r^H}\gls{Theta}\gls{h_i}\|}^2}{\gls{I}+\gls{sigma^2}}\label{gammaomaold},
\end{equation}
where \gls{I} is the interference power and \gls{sigma^2} is the noise variance.
In the case of \gls{IRS} assisted \gls{NOMA}, we consider two users multiplexed in power domain, with channel coefficients \gls{h_1} and \gls{h_2}. Additionally, we assume the following
\begin{align}
    \|\gls{h_r^H}\gls{Theta}\gls{h_1}\|^2 > \|\gls{h_r^H}\gls{Theta}\gls{h_2}\|^2.\nonumber
\end{align}
Thus, the \gls{BS} receives \gls{P_t}(\gls{alpha_1}\gls{s_1}+\gls{alpha_2}\gls{s_2}), where \gls{s_1} and \gls{s_2} are the data symbols transmitted by the strong and weak user, respectively. Further, we consider a fractional power control scenario, where $0<\gls{alpha_1}, \gls{alpha_2}\leq 1$ represent the fraction of the total available transmit power used by the strong and weak users, respectively. Thus, the signal received from the $i^{th}$ user in an \gls{IRS}-assisted \gls{NOMA} system is formulated as follows
\begin{equation}
  y_i^{\text{NOMA}}=\gls{h_r^H}\gls{Theta}\gls{h_i}\gls{P_t}(\gls{alpha_1}\gls{s_1}+\gls{alpha_2}\gls{s_2})+\gls{n_i}.
  \label{eqn:NOMASys}
\end{equation}
From \eqref{eqn:NOMASys}, we define the \gls{SINR} of the users in \gls{IRS}-assisted \gls{NOMA} as follows
\begin{align}
    \gls{gamma_1}&=\frac{\gls{alpha_1}\gls{P_t}{\|\gls{h_r^H}\gls{Theta}\gls{h_1}\|}^2}{\gls{alpha_2}\gls{P_t}{\|\gls{h_r^H}\gls{Theta}\gls{h_2}\|}^2+\gls{I}+\gls{sigma^2}}\label{gamma1old},\\
    \gls{gamma_2}&=\frac{{\gls{alpha_2}\gls{P_t}\|\gls{h_r^H}\gls{Theta}\gls{h_2}\|}^2}{\gls{I}+\gls{sigma^2}}\label{gamma2old}.
\end{align}

From~\eqref{eqn:hi}-\eqref{eqn:hr}, we define the following~\cite{base}:
\begin{align}
    \gls{h_r^H}\gls{Theta}\gls{h_i}&= \gls{alpha}\gls{beta}\mathbf{a}_M(\gls{psiBa},\gls{psiBe})\sum_{k=1}^N e^{j\gls{thetakcap}},\nonumber\\
    \|\mathbf{a}_M(\gls{psiBa},\gls{psiBe}) \|^2&=M,\nonumber\\
    \|\gls{h_r^H}\gls{Theta}\gls{h_i} \|^2 & ={|\gls{alpha}\gls{beta}|}^2M\left|\sum_{k=1}^N e^{j\gls{thetakcap}}\right|^2.\label{normsq}
\end{align}
We adopt the following approximation from~\cite{model}:
\begin{align}
    \left|\frac{1}{N}\sum_{k=1}^N e^{j\gls{thetakcap}}\right|^2\xrightarrow{(a)}\left|\mathbb{E}\left[e^{j\gls{thetakcap}}\right]\right|^2\overset{(b)}{=}\left|\mathbb{E}\left[\cos \gls{thetakcap}\right]\right|^2\overset{(c)}{=}\sinc^2(\gls{delta}),\label{sum}
\end{align}
where $(a)$ follows the strong law of large numbers~\cite{model}, in $(b)$, the expectation of odd function $\sin{\gls{thetakcap}}$ vanishes over the interval $\gls{thetakcap}\in[-\gls{delta}, \gls{delta}]$, and  $(c)$ is obtained by using the probability density function $f(\gls{thetakcap})=\frac{1}{2\gls{delta}}$, where $\gls{thetakcap}\in[-\gls{delta}, \gls{delta}]$.
Substituting \eqref{sum} in \eqref{normsq}
\begin{align}
    \|\gls{h_r^H}\gls{Theta}\gls{h_i} \|^2&={|\gls{alpha}\gls{beta}|}^2MN^2\sinc^2(\gls{delta}),\nonumber\\
    \|\gls{h_r^H}\gls{theta}\gls{h_i} \|^2&={|\gls{alpha}\gls{beta}|}^2MN^2.\nonumber
\end{align}
Using these approximations, we define the \gls{CSI} of $i^{th}$ user as
\begin{align}
    \gls{csi_i}&\triangleq \frac{\gls{P_t}\|\gls{h_r^H}\gls{theta}\gls{h_i} \|^2}{\gls{I}+\gls{sigma^2}}=\frac{\gls{P_t}|\gls{alpha}\gls{beta}|^2N^2M}{\gls{I}+\gls{sigma^2}}.\label{gammacsi}
\end{align}
Using \eqref{gammacsi} in \eqref{gammaomaold}, \eqref{gamma1old}, and  \eqref{gamma2old}, we formulate the received \gls{SINR}s as follows
\begin{align}
 \gls{gammaoma}&= \gls{csi_i}\sinc^2(\gls{delta}),\nonumber\\
 \gls{gamma_1}&=\frac{\gls{alpha_1}\gls{csi_1}\sinc^2(\gls{delta})}{1+\gls{alpha_2}\gls{csi_2}\sinc^2(\gls{delta})},\nonumber \\
 \gls{gamma_2}&=\gls{alpha_2}\gls{csi_2}\sinc^2(\gls{delta}).\nonumber
\end{align}
Given a logarithmic rate model, the normalized achievable data rates by the users in \gls{OMA} and \gls{NOMA} are  formulated as~\cite{mouni}
\begin{align}
    R_i^{\text{OMA}}&=\frac{1}{2}\log_2\left(1+\gls{csi_i}\sinc^2(\delta)\right),\label{eqn:ro}\\
    \gls{R_1}&=\log_2\left(1+\frac{\gls{alpha_1}\gls{csi_1}\sinc^2(\gls{delta})}{1+\gls{alpha_2}\gls{csi_2}\sinc^2(\gls{delta})}\right),\label{eqn:r1}\\
    R_2^{\text{NOMA}}&=\log_2\left(1+\gls{alpha_2}\gls{csi_2}\sinc^2(\gls{delta})\right)\label{eqn:r2}.
\end{align}

Next, we propose various adaptive user pairing algorithms.

\section{Proposed Algorithms}
\label{sec:algo}

In this section,  we propose \gls{MPA} and \gls{EEPA} algorithms to maximize the sum rate and energy efficiency, respectively. 
Additionally, in both the algorithms, we derive a criterion for pairing the users in \gls{NOMA} which ensures the minimum required data rates are achieved for each user. We denote \gls{R1bar} and \gls{R2bar} as the minimum required data rates by the strong and weak users, respectively. Thus, both \gls{MPA} and \gls{EEPA} algorithms should satisfy the following constraints
\begin{align}
    \gls{R_1}\geq\gls{R1bar},\label{eqn:r1constraint}\\
    \gls{R_2}\geq\gls{R2bar}.\label{eqn:r2constraint}
\end{align}
\vspace{-0.9cm}
\subsection{\gls{MPA}}
\label{subsec:mpa}
In \gls{MPA}, apart from satisfying \eqref{eqn:r1constraint}-\eqref{eqn:r2constraint}, we allocate the transmit powers (\gls{alpha_1} and \gls{alpha_2}) that maximize the \gls{ASR} of the paired users, where, \gls{ASR} is defined as
\begin{align}
    \gls{ASR}=\gls{R_1}+\gls{R_2}\nonumber.
\end{align}
We formulate the desired optimization as follows.
\begin{align}
\max_{\gls{alpha_1},\gls{alpha_2}} \hspace{0.4cm} &\gls{R_1}+\gls{R_2},\label{prob:mpa}\\
\text{s.t.}\hspace{0.3cm} & \eqref{eqn:r1constraint}, \  \eqref{eqn:r2constraint} \nonumber,\\
 & \gls{alpha_1}\ge0, \gls{alpha_2}\ge0,\label{eqn:0constr}\\
        & \gls{alpha_1}\le1, \gls{alpha_2}\le1.\label{eqn:1constr} 
\end{align}

\subsubsection{Bounds on $\gls{alpha_1}$ and $\gls{alpha_2}$}
\label{subsec:mpabounds}
From \eqref{eqn:r1}-\eqref{eqn:r2}, we get, \begin{align}
\gls{R_1}+\gls{R_2}=\log_2(1+(\gls{alpha_1}\gls{csi_1}+\gls{alpha_2}\gls{csi_2})\sinc^2(\gls{delta})).\label{eqn:ASR}
\end{align}
Thus, \gls{ASR} is an increasing function with respect to \gls{alpha_1} and \gls{alpha_2}. However, from \eqref{eqn:r1}, an increase in \gls{alpha_2} will decrease the achievable data rate for the strong user, whereas, there is no such impact with an increase in \gls{alpha_1}. Hence, to maximize the \gls{ASR}, we assign $\gls{alpha_1}=1$. Thus, we get,
\begin{align}
    \gls{R_1}&=\log_2\left(1+\frac{\gls{csi_1}\sinc^2(\gls{delta})}{1+\gls{alpha_2}\gls{csi_2}\sinc^2(\gls{delta})}\right).\label{eqn:r11}
    \end{align}
From \eqref{eqn:r2} and \eqref{eqn:r2constraint},  we get,
\begin{align}
    \gls{alpha_2}\ge\frac{2^{\gls{R2bar}}-1}{\gls{csi_2}\sinc^2(\gls{delta})}&\triangleq \gls{alpha2lb}.\label{eqn:alpha2lb}
\end{align}
From \eqref{eqn:r1constraint} and \eqref{eqn:r11}, we get,

\begin{align}
\gls{alpha_2}\leq \frac{\gls{csi_1}\sinc^2(\gls{delta})+1-2^{\gls{R1bar}}}{\gls{csi_2}\sinc^2(\gls{delta})[2^{\gls{R1bar}}-1]}&\triangleq \gls{alpha2ub}.\label{eqn:alpha2ub}
\end{align}
\begin{figure}[t]
    \centering
    \includegraphics[scale=0.5]{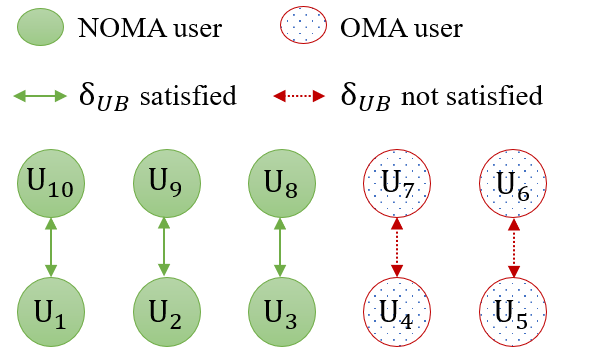}
    \caption{Illustration of adaptive user pairing.}
    \label{fig:users}
\vspace{-0.4cm}
\end{figure}\vspace{-0.2cm}
\subsubsection{Criterion for pairing the users}
Using \eqref{eqn:alpha2lb}-\eqref{eqn:alpha2ub}, and assuming
$\gls{alpha2ub}\geq\gls{alpha2lb}$, we obtain the pairing criterion for the \gls{MPA} algorithm as follows.
\begin{align}
    \sinc^2(\gls{delta})\geq\frac{2^{\gls{R2bar}}\left[2^{\gls{R1bar}}-1\right]}{\gls{csi_1}}\triangleq \sinc^2(\gls{delubmpa}).\label{eqn:delubmpa}
\end{align}
We define \gls{R1bar} and \gls{R2bar} as the achievable \gls{OMA} rates, and then, pair the users in \gls{NOMA} iff \eqref{eqn:delubmpa} is satisfied. Otherwise, we consider transmitting the information for the users in an \gls{OMA} scenario as shown in Fig.~\ref{fig:users}. This way, the proposed algorithm ensures that a minimum of \gls{OMA} rates are achieved in a worst-case scenario and it maximizes the \gls{ASR} by switching to \gls{NOMA} whenever feasible.

\subsubsection{Power allocation}
For the paired users, we define the power allocation factors as follows.
\begin{align}
    \gls{alpha1mpa}&=1,\label{eqn:alpha1mpa}\\ 
    \gls{alpha2mpa}&=\min\{\gls{alpha2ub},1\}.\label{eqn:alpha2mpa}
\end{align}
Note that in \eqref{eqn:alpha1mpa}, the strong user transmits at maximum power to maximize the spectral efficiency,  
whereas, in \eqref{eqn:alpha2mpa}, the weak user transmits at a maximum possible power that does not degrade the strong user's data rate beyond the \gls{OMA} rate. Further, to avoid transmit power violations, we limit the maximum value of \gls{alpha2mpa} to 1 in \eqref{eqn:alpha2mpa}. 
\begin{lemma}
The power allocation factors formulated in {\eqref{eqn:alpha1mpa}~-~\eqref{eqn:alpha2mpa}} are the optimal values that achieve maximum sum rate while ensuring the individual \gls{NOMA} rates to be better than the \gls{OMA} counterparts.
\end{lemma}

\begin{proof}
For the ease of understanding, we define
\begin{align}
    \gls{eta}&\triangleq\frac{2^{\gls{R1bar}}-1}{\gls{csi_1}\sinc^2(\gls{delta})},\label{eqn:eta}\\
    \gls{kappa}&\triangleq\frac{(2^{\gls{R1bar}}-1)\gls{csi_2}}{\gls{csi_1}}.\label{eqn:kappa}
\end{align}
Thus, we reformulate \eqref{prob:mpa} as follows
\begin{align}
 \max_{\gls{alpha_1},\gls{alpha_2}} \hspace{0.45cm} &\gls{alpha_1}\gls{csi_1}+\gls{alpha_2}\gls{csi_2},\label{prob:finalmpa}\\
\text{s.t.}\hspace{0.35cm} 
&-\gls{alpha_1}+\gls{alpha_2}\gls{kappa}+\gls{eta}\le0,\label{eqn:Lemm3}\\
& -\gls{alpha_2}+\gls{alpha2lb}\le0,\label{eqn:Lemm2}\\
& \hspace{0.4cm} \gls{alpha_1}-1\le0, \gls{alpha_2}-1\le0.\label{eqn:Lemm1}
\end{align}
where, \eqref{prob:finalmpa} is obtained by using the fact that logarithmic function is a monotonically increasing function and substituting \eqref{eqn:ASR} in \eqref{prob:mpa}. The constraint \eqref{eqn:Lemm3} is obtained by substituting \eqref{eqn:eta}-\eqref{eqn:kappa} in \eqref{eqn:r1constraint} and the constraint \eqref{eqn:Lemm2} is obtained by further solving the \eqref{eqn:r2constraint}. Additionally, since \gls{R1bar}, \gls{R2bar} are non negative, we get $\gls{alpha2lb},\ \gls{eta},\ \gls{kappa}>0$, and thus, we consider the constraint \eqref{eqn:0constr} is already captured in \eqref{eqn:Lemm3}-\eqref{eqn:Lemm2}. Next, the Lagrangian for \eqref{prob:finalmpa} is formulated as
\begin{align}
    L(\gls{alpha_1},\gls{alpha_2},\mu_1,\mu_2,\mu_3,\mu_4)=\gls{alpha_1}\gls{csi_1}+\gls{alpha_2}\gls{csi_2}-\mu_1(\gls{alpha_1}-1)\nonumber\\-\mu_2(\gls{alpha_2}-1)-\mu_3(-\gls{alpha_2}+\gls{alpha2lb})-\mu_4(-\gls{alpha_1}+\gls{alpha_2}\gls{kappa}+\gls{eta}).\nonumber
\end{align}
Solving stationarity conditions $\frac{\partial L}{\partial \gls{alpha_1}}=0\  \text{and}\  \frac{\partial L}{\partial \gls{alpha_2}}=0$, we get
\begin{align}
    \gls{csi_1}-\mu_1+\mu_4&=0,\label{eqn:kkt1}\\
    \gls{csi_2}-\mu_2+\mu_3-\mu_4\gls{kappa}&=0.\label{eqn:kkt2}
\end{align}
The complementary slackness conditions are formulated as
\begin{align}
   \mu_1(\gls{alpha_1}-1)&=0,\label{eqn:kkt3}\\
   \mu_2(\gls{alpha_2}-1)&=0,\label{eqn:kkt4}\\
   \mu_3(-\gls{alpha_2}+\gls{alpha2lb})&=0,\label{eqn:kkt5}\\
   \mu_4(-\gls{alpha_1}+\gls{alpha_2}\gls{kappa}+\gls{eta})&=0.\label{eqn:kkt6}
\end{align}
The dual feasibility conditions are formulated as
\begin{align}
    \mu_i\ge0\hspace{0.1cm},\ \forall\hspace{0.1cm}i \in [1,4].\label{eqn:kkt7}
\end{align}
Solving \eqref{eqn:kkt1}-\eqref{eqn:kkt7}, we find that the possible optimal values of $(\gls{alpha_1},\gls{alpha_2})$ are
$(\gls{alpha2lb}\gls{kappa}+\gls{eta},\ \gls{alpha2lb}), \ (\gls{kappa}+\gls{eta},1),\ (1,1),\ (1,\gls{alpha2lb}),$ and $\ (1,\frac{1-\gls{eta}}{\gls{kappa}})$. From \eqref{prob:finalmpa}, larger the values of \gls{alpha_1} and \gls{alpha_2}, larger will be the \gls{ASR}. Since, $\gls{alpha_1}=1$ has no impact on any of the desired constraints, we consider only the solutions with $\gls{alpha_1}=1$ which are 
$(1,1), (1,\gls{alpha2lb})$, and $(1,\frac{1-\gls{eta}}{\gls{kappa}})$. Note that considering $(\gls{alpha_1},\gls{alpha_2})=(1,1)$ will violate the constraint \eqref{eqn:Lemm3}. Hence, by substituting $\gls{alpha2ub}=\frac{1-\gls{eta}}{\gls{kappa}}$, the optimal values of (\gls{alpha_1},\gls{alpha_2}) are either  $(1,\gls{alpha2ub})$ or $(1,\gls{alpha2lb})$. Since, $\gls{alpha2ub}\geq \gls{alpha2lb}$, we conclude $(\gls{alpha_1},\gls{alpha_2})=(1,\gls{alpha2ub})$ as the optimal solution.
This completes the proof of the Lemma 1.
\end{proof}
Next, we present the \gls{EEPA} algorithm.
\begin{algorithm}[t]
\caption{Proposed algorithms} \label{alg:aup}
\kwin{$\gls{csi_i},\  \forall i \in [1,G]$.}
\kwvar{$i$ represents the user pairing index.}
Sort the $G$ users in decreasing order of \gls{csi_i}\;
Set $i=1$\;
\While{$i<\frac{G}{2}+1$}{
{Consider $i^{th}$ user as strong user and $(G-i+1)^{th}$ user as the weak user\;
Calculate \gls{R1bar}, \gls{R2bar}} from \eqref{eqn:ro}\;
\uIf{\gls{MPA}}{
Calculate \gls{delubmpa} from \eqref{eqn:delubmpa}\;
\eIf{$\gls{delta}\le$ \gls{delubmpa}}{
Pair the users in \gls{NOMA} with \gls{alpha_1}=1, \gls{alpha_2}=$\min\{\gls{alpha2ub},1\}$ as per \eqref{eqn:alpha1mpa}, \eqref{eqn:alpha2mpa}}
{Consider the users in \gls{OMA}\;}
}
\uElseIf{\gls{EEPA}}{
Calculate $\gls{delubeepa}=\min\{\gls{delub1},\gls{delub2}\}$ 
from \eqref{eqn:delubeepa}\;
\eIf{$\gls{delta}\le\gls{delubeepa}$}{Pair the users in \gls{NOMA} with (\gls{alpha_1},\ \gls{alpha_2}) obtained from the solution of \eqref{prob:eepa}\;}
{Consider the users in \gls{OMA}\;}
}$i=i+1$\;}

\end{algorithm}



\begin{figure*}
    \centering
    \begin{subfigure}[t]{0.25\textwidth}
    \centering
    \includegraphics[width=\textwidth]{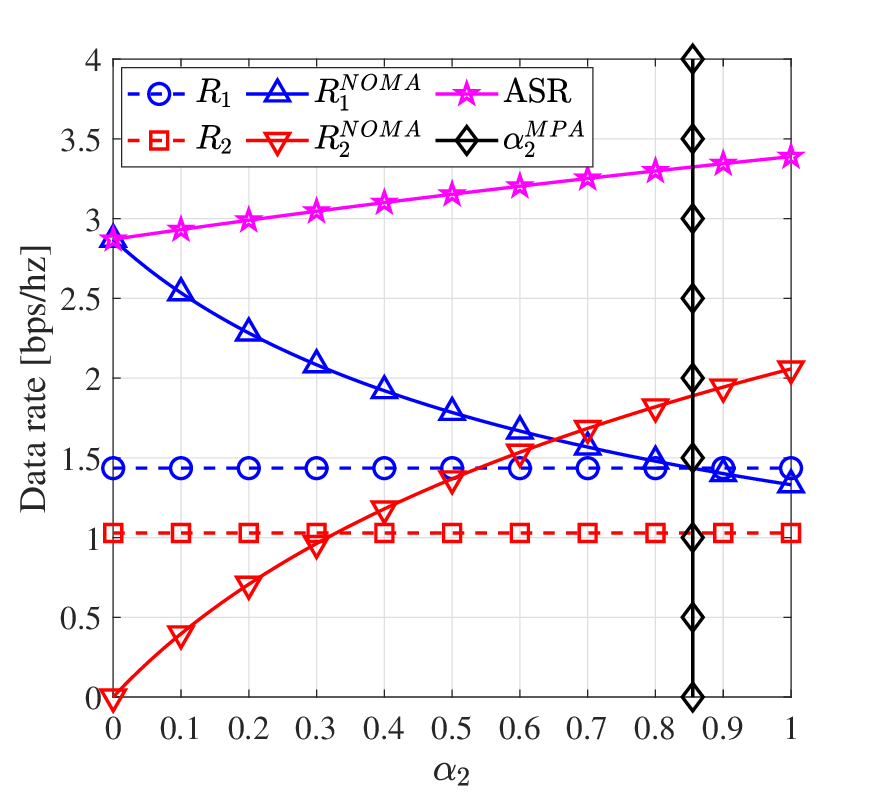}
    \captionsetup{font=footnotesize}
    \caption{\gls{csi_1},\gls{csi_2}=[8,5] dB, \gls{delta}=$0\degree$.}
    \label{subfig:8,5,0}
    \end{subfigure}%
    \begin{subfigure}[t]{0.25\textwidth}
    \centering
    \includegraphics[width=\textwidth]{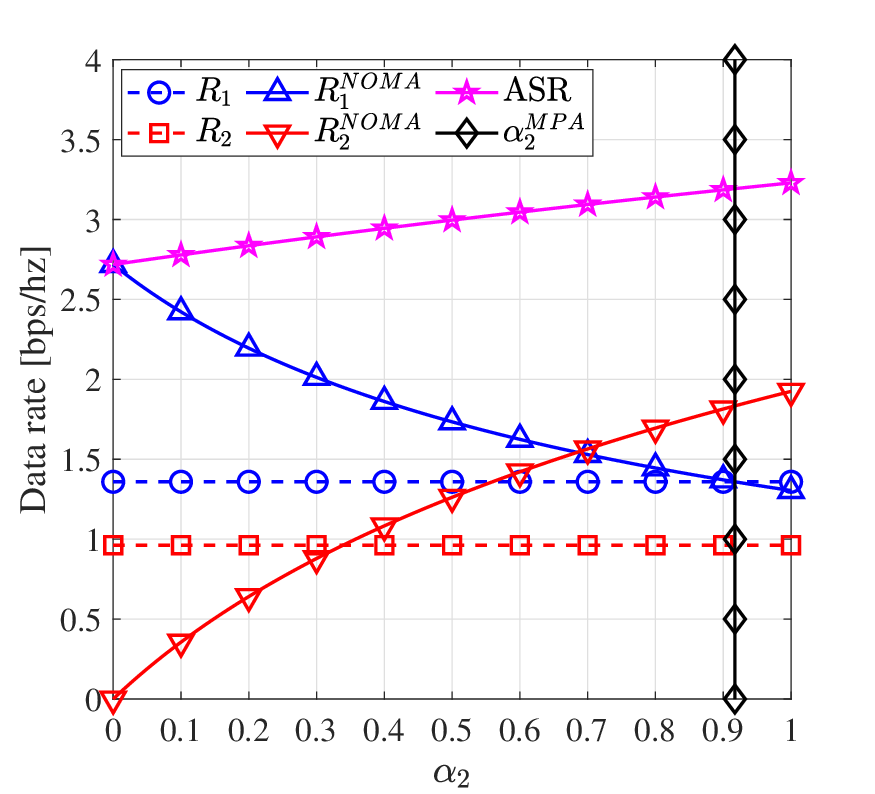}
    \captionsetup{font=footnotesize}
    \caption{\gls{csi_1},\gls{csi_2}=[8,5] dB, \gls{delta}=$11\degree$.}
    \label{subfig:8,5,11}
    \end{subfigure}%
    \begin{subfigure}[t]{0.25\textwidth}
    \centering
    \includegraphics[width=\textwidth]{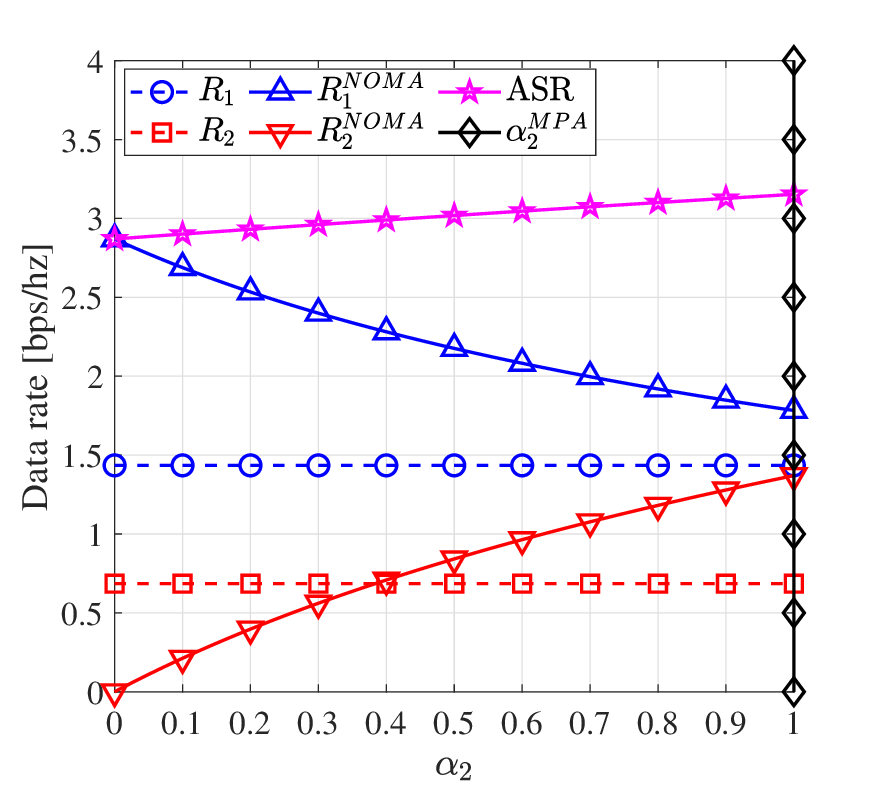}
    \captionsetup{font=footnotesize}
    \caption{\gls{csi_1},\gls{csi_2}=[8,2] dB, \gls{delta}=$0\degree$.}
    \label{subfig:8,2,0}
    \end{subfigure}%
    \begin{subfigure}[t]{0.25\textwidth}
    \centering
    \includegraphics[width=\textwidth]{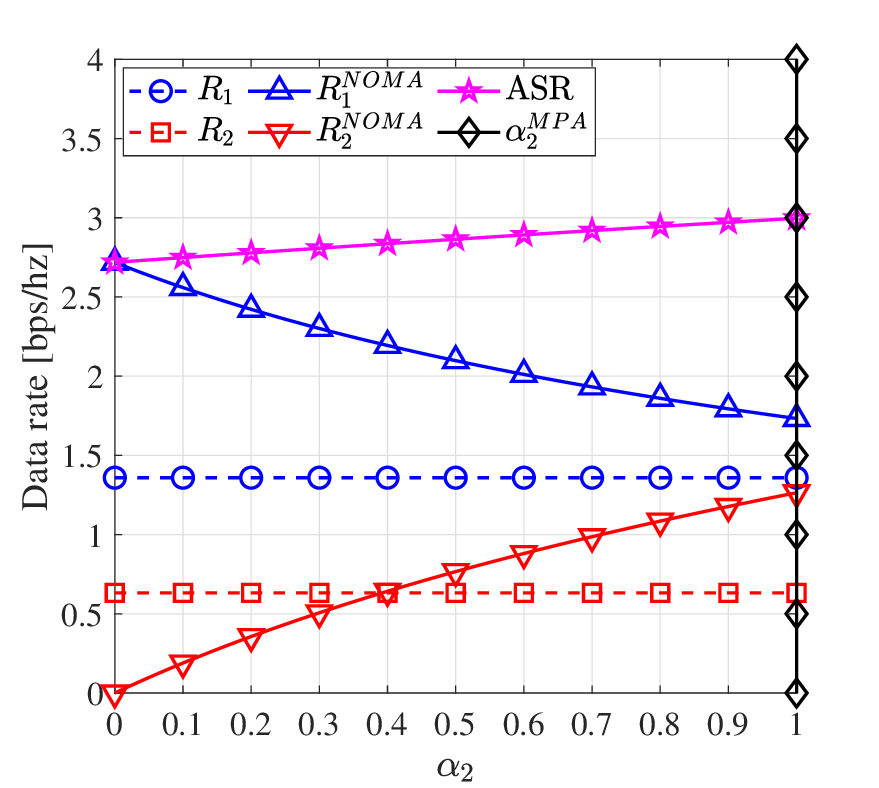}
    \captionsetup{font=footnotesize}
    \caption{\gls{csi_1},\gls{csi_2}=[8,2] dB, \gls{delta}=$11\degree$.}
    \label{subfig:8,2,11}
    \end{subfigure}
    \caption{Comparison of achievable data rates for varying power allocation factor of the weak user.}
    \label{Fig:ratesvsalpha2}
    \end{figure*}

\subsection{EEPA}
In \gls{EEPA}, apart from satisfying \eqref{eqn:r1constraint}-\eqref{eqn:r2constraint}, we allocate the powers (\gls{alpha_1} and \gls{alpha_2}) that maximize the \gls{EE} of the paired users, where, \gls{EE} is defined as
\vspace{-1.4cm}
\begin{align}
    \gls{EE}&=\frac{\gls{ASR}}{\gls{alpha_1}+\gls{alpha_2}},\nonumber\\&=\frac{\log_2\left[1+\left(\gls{alpha_1}\gls{csi_1}+\gls{alpha_2}\gls{csi_2}\right)\sinc^2(\gls{delta})\right]}{\gls{alpha_1}+\gls{alpha_2}}.\nonumber
\end{align}
Thus, we formulate the optimisation problem as follows.
\begin{align}
        \underset{\gls{alpha_1},\gls{alpha_2}}{\max } \hspace{0.4cm} &  \frac{\log_2\left[1+\left(\gls{alpha_1}\gls{csi_1}+\gls{alpha_2}\gls{csi_2}\right)\sinc^2(\gls{delta})\right]}{\gls{alpha_1}+\gls{alpha_2}}\label{prob:eepa},\\
        \text{s.t.}\hspace{0.3cm} & \eqref{eqn:Lemm3}-\eqref{eqn:Lemm1}.\nonumber
\end{align}\vspace{-0.9cm}
\subsubsection{Criterion for pairing the users}
\label{subsec:eepabounds}

Solving \eqref{eqn:Lemm3} for a worst case scenario of \gls{alpha_2} $=1$,  we get,
\begin{align}
    \gls{alpha_1} \ge \gls{kappa}+\gls{eta}&=\frac{2^{\gls{R1bar}} -1}{\gls{csi_1}}\left[\gls{csi_2}+\frac{1}{\sinc^2(\gls{delta})}\right].\label{eqn:alpha1ub}
\end{align} 
Using \eqref{eqn:alpha1ub} in \eqref{eqn:1constr}, we get,
\begin{align}
    \frac{2^{\gls{R1bar}} -1}{\gls{csi_1}} \left[\gls{csi_2}+\frac{1}{\sinc^2(\gls{delta})}\right]\le&1,\nonumber\\
     \sinc^2(\gls{delta})\ge&\frac{1}{\left(\frac{\gls{csi_1}}{2^{\gls{R1bar}} -1}\right)-\gls{csi_2}}\triangleq \sinc^2(\gls{delub1}).
    \nonumber
\end{align}
Using \eqref{eqn:alpha2lb} in \eqref{eqn:1constr}, we get, 
\begin{align}
    \gls{alpha2lb}&\le 1, \nonumber\\
    \sinc^2(\gls{delta})&\ge \frac{2^{\gls{R2bar}}-1}{\gls{csi_2}} \triangleq \sinc^2(\gls{delub2}). \nonumber 
\end{align}
Thus, we define the upper bound on the phase imperfection in \gls{EEPA} as
\begin{align}
    \gls{delubeepa}=\min\{\gls{delub1},\gls{delub2}\}.\label{eqn:delubeepa}
\end{align}
We define \gls{R1bar} and \gls{R2bar} as the achievable \gls{OMA}  rates, and then, pair the users in \gls{NOMA} iff the imperfection in the phase compensation is less than \gls{delubeepa}. Otherwise, we consider transmitting the information for the users in an \gls{OMA} scenario as shown in Fig.~\ref{fig:users}. 

\subsubsection{Power allocation}
The objective function formulated in \eqref{prob:eepa} is a strictly pseudo-concave function~\cite{prop6} and an efficient way of obtaining a solution is to use the Dinkelbach's algorithm~\cite{prop6,eepa}.  An outline of implementing \gls{MPA} and \gls{EEPA} is presented in detail in Algorithm~\ref{alg:aup}.

\vspace{-0.3cm}

\section{Numerical Results}
\label{sec:NR}
\begin{figure}[t]
    \centering
    \vspace{-0.4cm}
    \begin{subfigure}[t]{0.26\textwidth}
    \centering
    \includegraphics[width=\textwidth]{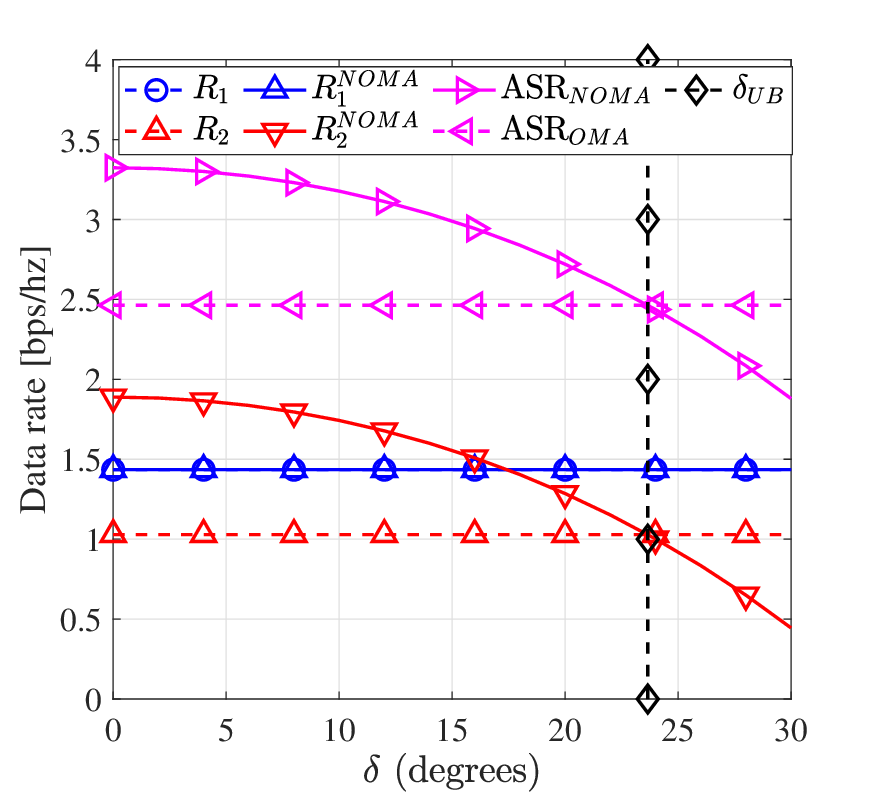}
    \captionsetup{font=footnotesize}
    \caption{\gls{csi_1},\gls{csi_2}=[8,5] dB.}
    \label{subfig:ratevsdelta8,5}
    \end{subfigure}%
    \begin{subfigure}[t]{0.26\textwidth}
    \centering
    \includegraphics[width=\textwidth]{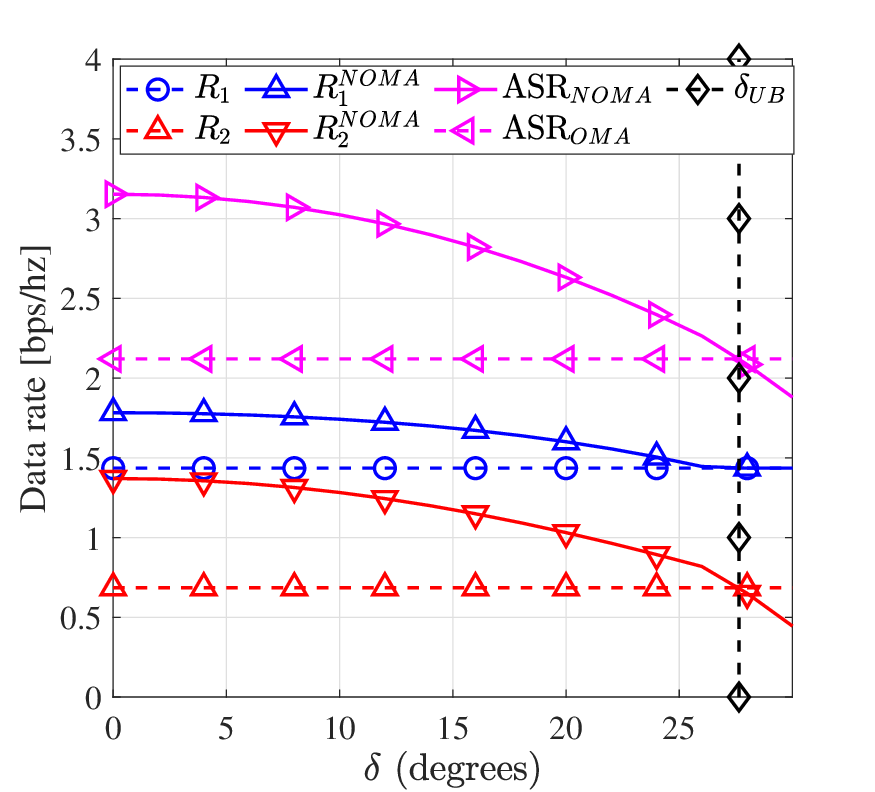}
    \captionsetup{font=footnotesize}
    \caption{\gls{csi_1},\gls{csi_2}=[8,2] dB.}
    \label{subfig:ratevsdelta8,2}
    \end{subfigure}
    \caption{Comparison of achievable data rates for varying imperfection in phase.}
    \label{Fig:ratesvsdelta}\vspace{-0.4cm}
\end{figure}
For the simulations, we have considered $M=8$, $N=32$, and $\overline{R_i}=R_i^\text{OMA}$. We have then dropped the \gls{BS}s and users from Poisson point distribution with densities of 25 BS/km$^2$ and 2000 users/km$^2$, respectively. For each user, we have calculated the path loss and the received signal power from each \gls{BS} by considering the urban cellular path loss model presented in \cite{38901}.
The users are then associated to the \gls{BS} from which they receive the maximum signal power, and the rest of the \gls{BS}s are considered as interfering \gls{BS}s. Given this simulation set-up, we have calculated various performance metrics with the proposed and the existing state-of-the-art algorithms which are summarised as follows.

In Fig.~\ref{Fig:ratesvsalpha2}, we present the achievable data rates for varying \gls{alpha2mpa}. For the evaluation, we consider two configurations of user pairs with $[\gls{csi_1},\gls{csi_2}]=[8,5]~\text{dB}$, $[\gls{csi_1},\gls{csi_2}]=[8,2]~\text{dB}$, and two configurations of imperfection in phase compensation, $\gls{delta}=0^\circ, 11^\circ$. As shown in Fig.~\ref{Fig:ratesvsalpha2}, with increasing $\alpha_2$, the data rates for weak user increases. This increase in $\alpha_2$ also increases interference for the strong user, and thus, the achievable data rate for the strong user decreases. Note that for a fixed \gls{delta}, the achievable data rates and sum rate are better in Fig.~\ref{subfig:8,5,0} as compared to Fig.~\ref{subfig:8,2,0} because of the better \gls{SINR} conditions. Further, with an increase in the \gls{delta}, the achievable data rates and sum rate are smaller in Fig.~\ref{subfig:8,5,11} when compared to Fig.~\ref{subfig:8,5,0}. As shown in Fig.~\ref{Fig:ratesvsalpha2}, the \gls{ASR} is a non-decreasing function of \gls{alpha_2} which aligns with our formulation in the Section~\ref{subsec:mpa}. Further, beyond the proposed $\alpha_2=\gls{alpha2mpa}$, the individual data rates are not better than the \gls{OMA} counterparts. Thus, we validate the proposed bounds on power allocation factors in the presence of imperfect phase compensation.

In Fig.~\ref{Fig:ratesvsdelta}, we present  the comparison of the achievable data rates for varying \gls{delta}. 
For the evaluation, we consider two configurations of user pairs with $[\gls{csi_1},\gls{csi_2}]=[8,5]$~dB and $[\gls{csi_1},\gls{csi_2}]=[8,2]$~dB in Fig.~\ref{subfig:ratevsdelta8,5} and Fig.~\ref{subfig:ratevsdelta8,2},  respectively.
Further, we consider $\alpha_1=1$ and $\alpha_2=\gls{alpha2mpa}$ while calculating the achievable data rates. With an increase in \gls{delta}, the achievable rates decrease. Additionally, whenever $\gls{delta}<\gls{delubmpa}$, the achievable data rates and sum-rates are always better than the \gls{OMA} counterparts. Thus, we validate the proposed bound on the imperfection in the phase compensation. Note that a similar analysis is extendable for the \gls{EEPA} scenario.

    \begin{figure}\vspace{-0.4cm}
    \begin{subfigure}[t]{0.26\textwidth}
    \centering
    \includegraphics[width=\textwidth]{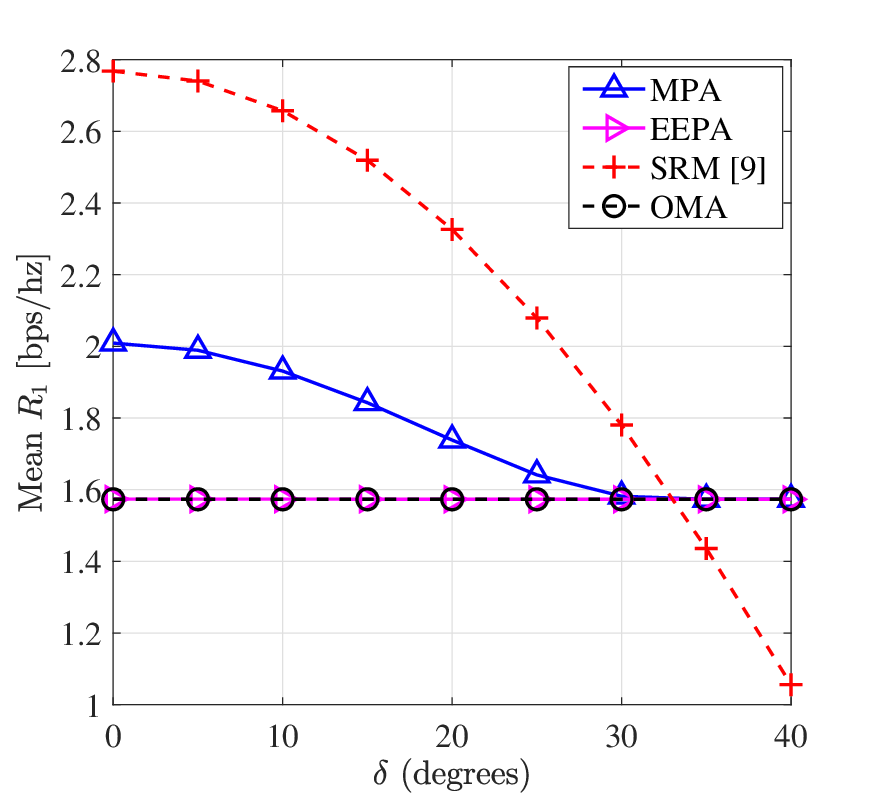}
    \captionsetup{font=footnotesize}
    \caption{Mean $R_1$ vs \gls{delta}.}
    \label{subfig:r1vsdelta}
    \end{subfigure}%
    \begin{subfigure}[t]{0.26\textwidth}
    \centering
    \includegraphics[width=\textwidth]{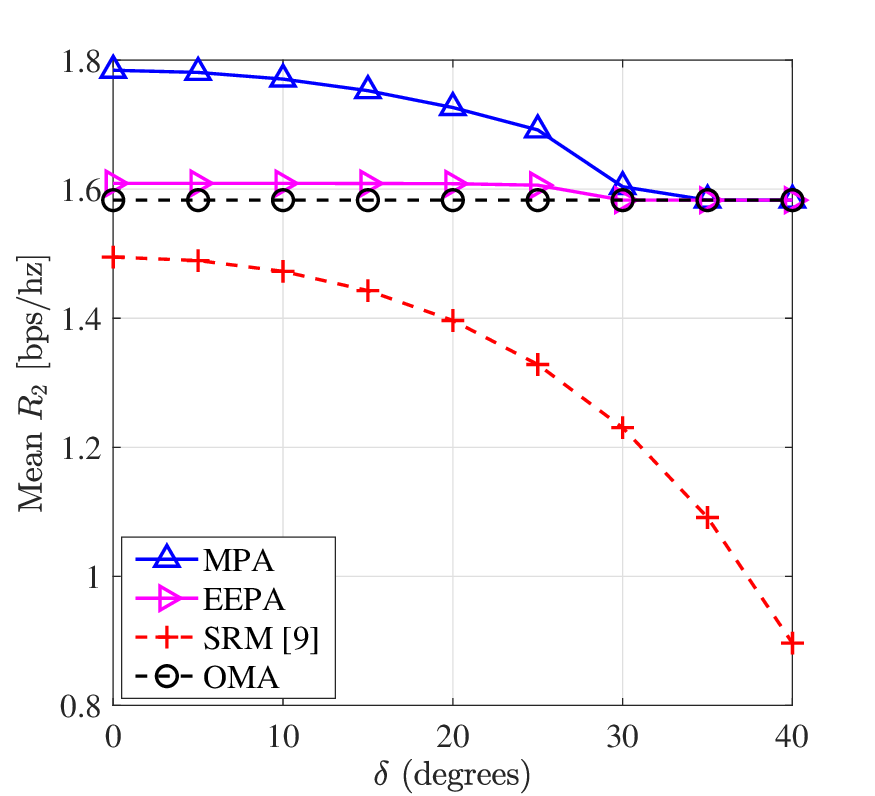}
    \captionsetup{font=footnotesize}
    \caption{Mean $R_2$ vs \gls{delta}.}
    \label{subfig:r2vsdelta}
    \end{subfigure}
    \caption{Comparison of mean achievable data rates of strong and weak user with various algorithms.}
    \label{Fig:comp1}
    \vspace{-0.4cm}\end{figure}
    
In Fig.~\ref{Fig:comp1}, we present the performance comparison of mean of  achievable data rates with various algorithms. As shown in Fig.~\ref{Fig:comp1}, for both strong and weak user, the performance of \gls{SRM}~\cite{srm} declines gradually with increase in \gls{delta} and the data rates fall below \gls{OMA} for larger \gls{delta}. In case of the proposed algorithms, both \gls{MPA} and \gls{EEPA} consider the imperfections in the phase compensation, and hence, the data rates gradually converge to \gls{OMA} the rates with an increase in \gls{delta}. 
Further, as shown in Fig.~\ref{subfig:r2vsdelta}, the \gls{SRM} algorithm tries to maximize the \gls{ASR} and in the process significantly decreases the weak user data rates beyond the required \gls{OMA} rates. However, the \gls{MPA} algorithm maximizes the strong user rate to achieve maximum \gls{ASR} and yet ensures both strong and weak user achieve minimum of \gls{OMA} rates. The \gls{EEPA} algorithm allocates minimum power to each user to ensure the minimum required \gls{OMA} rates are achieved, and hence, the data rates with \gls{EEPA} are significantly lower than \gls{MPA} rates and slightly higher than the \gls{OMA} rates.

    \begin{figure}
    \centering
    \includegraphics[width=0.38\textwidth]{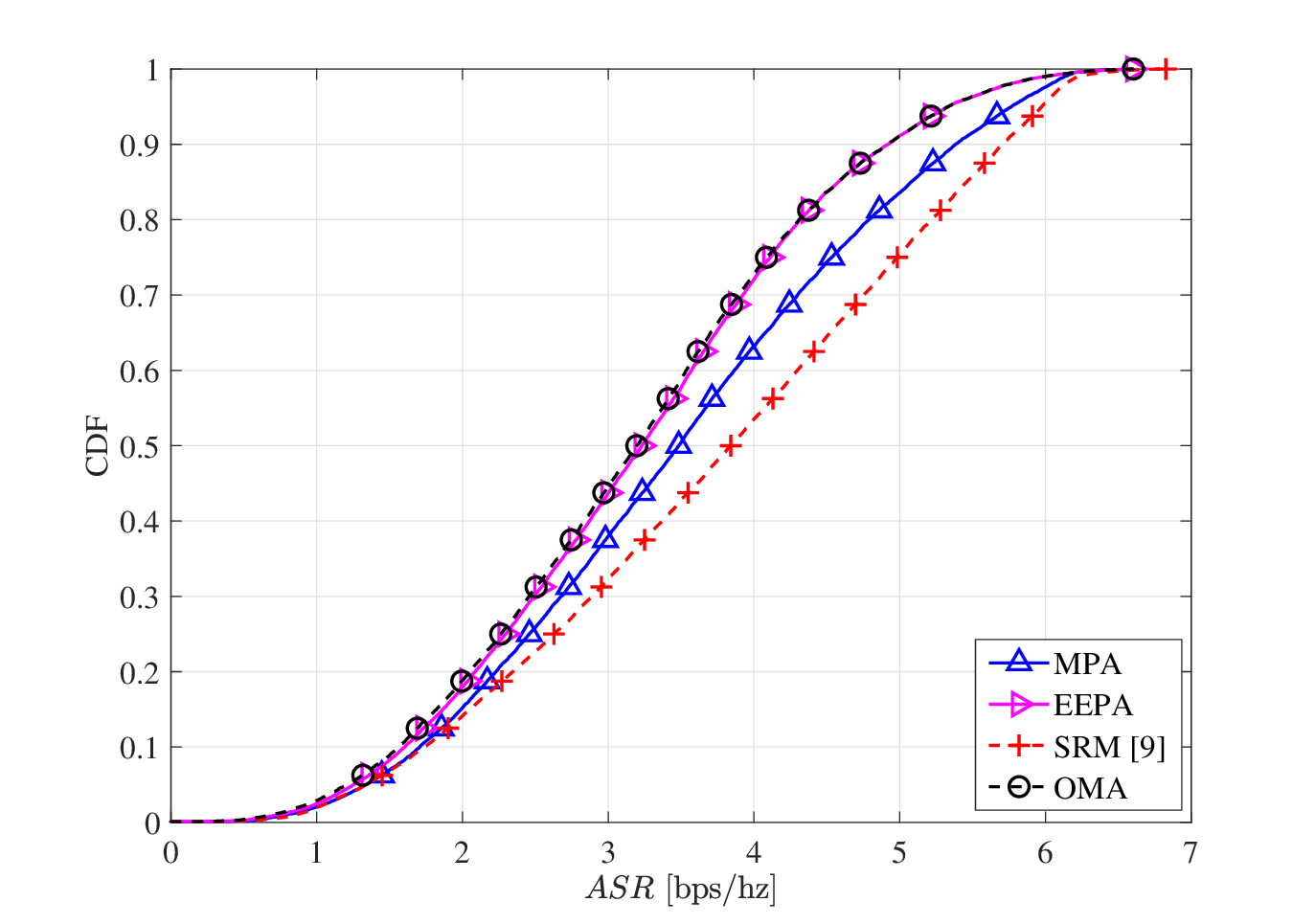}
    \caption{Comparison of mean of achievable data rates of strong and weak user with various algorithms.}
    \vspace{-0.5cm}
    \label{Fig:Asrcdf}
    \end{figure}

In Fig.~\ref{Fig:Asrcdf}, we present the comparison of the mean \gls{ASR} with the proposed algorithms against the \gls{SRM} and \gls{OMA} in the presence of imperfect phase compensation. As shown in Fig.~\ref{Fig:Asrcdf}, \gls{SRM} achieves highest \gls{ASR} as compared to all the algorithms.  However, as shown in Fig.~\ref{Fig:comp1}, the \gls{SRM} does not achieve minimum required rates for the weak user,
whereas, the proposed \gls{MPA} maximizes \gls{ASR} while ensuring a minimum of \gls{OMA} rates for both strong and weak users. Further, note that with \gls{EEPA}, the achievable sum rate is slightly higher than the \gls{OMA} rates.


In Fig.~\ref{Fig:comp2}, we present the comparison of the mean achievable sum rate and energy efficiency with various algorithms.  As shown in Fig.~\ref{subfig:asrvsdelta} and \ref{subfig:eevsdelta}, the \gls{SRM} has highest mean \gls{ASR} and \gls{EE} at the lower \gls{delta}. However, note that \gls{SRM} algorithm does not ensure that individual users achieve a minimum of \gls{OMA} rates. 
Further, with increasing \gls{delta}, the mean \gls{ASR} and \gls{EE} of the both the proposed algorithms converge to the \gls{OMA} rates, whereas, the performance of the  \gls{SRM} degrades significantly as compared to the \gls{OMA} rates.  Hence, it is not always beneficial to pair the users in \gls{NOMA}. Thus, we conclude that, the proposed algorithms outperform the existing algorithms in presence of imperfection in phase compensation. Additionally, they also maximize the data rates or energy efficiency while ensuring minimum required data rates for each user.

    \begin{figure}
    \centering
    \begin{subfigure}[t]{0.26\textwidth}
    \centering
    \includegraphics[width=\textwidth]{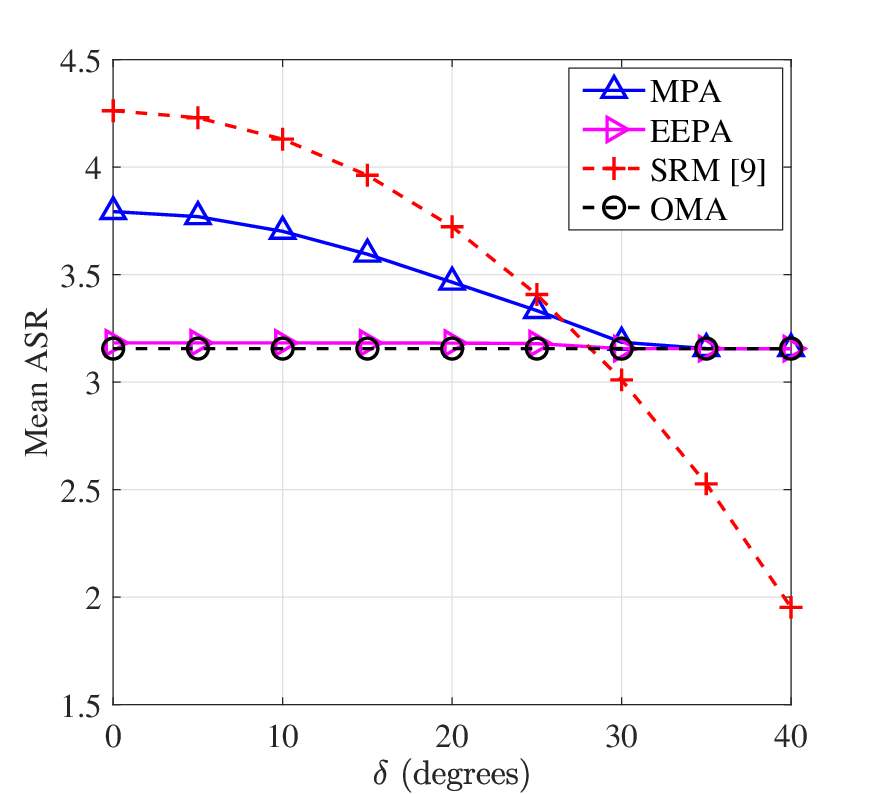}
    \captionsetup{font=footnotesize}
    \caption{Mean \gls{ASR} vs \gls{delta}.}
    \label{subfig:asrvsdelta}
    \end{subfigure}%
    \begin{subfigure}[t]{0.26\textwidth}
    \centering
    \includegraphics[width=\textwidth]{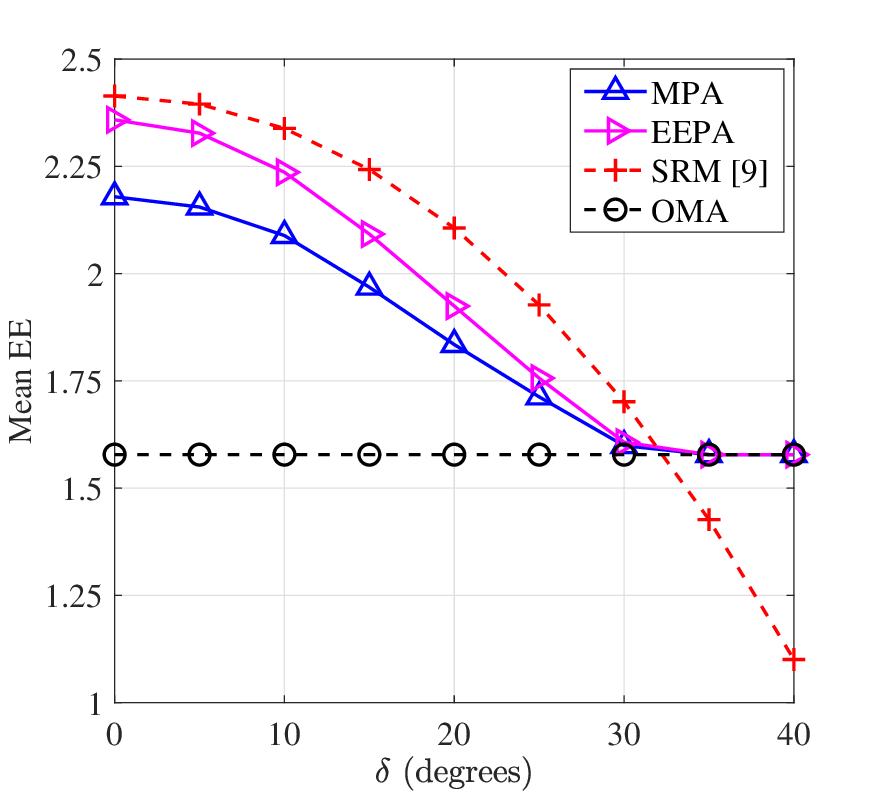}
    \captionsetup{font=footnotesize}
    \caption{Mean \gls{EE} vs \gls{delta}.}
    \label{subfig:eevsdelta}
    \end{subfigure}
    \caption{Comparison of mean achievable sum rate and energy efficiency with various algorithms.}
    \label{Fig:comp2}
    \end{figure}




\vspace{-0.1cm}
\section{Conclusion}
\label{sec:conclusion}
We have proposed adaptive user pairing algorithms for \gls{IRS}-assisted uplink \gls{NOMA} systems that maximize the achievable sum-rate or energy efficiency. We have formulated the criterion for user pairing based on the derived bounds on imperfection in the phase compensation. Through numerical results, we have validated the derived bounds and the proposed pairing criterion. Further, we have proposed novel power allocation procedures for the paired users. We have performed extensive system-level simulations and have shown that the proposed algorithms achieve significant improvement over the state-of-the-art algorithms with increase in phase imperfection. In the future, we plan to validate the proposed algorithms on the hardware test-beds.
\vspace{-0.5cm}
\bibliographystyle{ieeetran}
\bibliography{Bibfile}
\end{document}